\documentclass[
 prx,
 aps,
 reprint,
 superscriptaddress,
 longbibliography,
 dvipsnames,
 nofootinbib
]{revtex4-2}

\pdfoutput=1

\usepackage[utf8]{inputenc}
\usepackage{verbatim, mathtools, amsmath, amsfonts, amssymb, amsthm}
\usepackage{graphicx} 
\usepackage{xcolor}
\usepackage{hyperref}
\usepackage{xspace}

\usepackage{xstring} 
\usepackage{bm}
\usepackage{dsfont} 
\usepackage{textgreek} 
\usepackage[group-separator={,}, group-minimum-digits=3]{siunitx} 
\usepackage{booktabs} 
\usepackage{multirow}

\usepackage{adjustbox}

\newtheorem{definition}{Definition}
\newtheorem{theorem}{Theorem}
\newtheorem{corollary}{Corollary}[theorem]
\newtheorem{lemma}[theorem]{Lemma}

\usepackage{diagbox}
\usepackage{url}
\hypersetup{
    colorlinks=true,
    linkcolor=blue,
    filecolor=magenta,      
    urlcolor=cyan,
    pdftitle={Overleaf Example},
    pdfpagemode=FullScreen,
    }

\makeatletter
\def\@xfootnote[#1]{%
  \protected@xdef\@thefnmark{#1}%
  \@footnotemark\@footnotetext}
\makeatother

\usepackage[usestackEOL]{stackengine}

\begin{document}

\title{Explicit Instances of Quantum Tanner Codes}
\author{Rebecca Katharina Radebold}
\email{rrad0225@uni.sydney.edu.au}
\affiliation{Centre for Engineered Quantum Systems, School of Physics,
The University of Sydney, Sydney, New South Wales 2006, Australia}
\affiliation{Sydney Quantum Academy, Sydney, New South Wales, Australia}

\author{Stephen D. Bartlett}
\affiliation{Centre for Engineered Quantum Systems, School of Physics,
The University of Sydney, Sydney, New South Wales 2006, Australia}

\author{Andrew C. Doherty}
\affiliation{Centre for Engineered Quantum Systems, School of Physics,
The University of Sydney, Sydney, New South Wales 2006, Australia}

\begin{abstract}
We construct several explicit instances of quantum Tanner codes, a class of asymptotically good quantum low-density parity check (qLDPC) codes. The codes are constructed using dihedral groups and random pairs of classical codes and exhibit high encoding rates, relative distances, and pseudo-thresholds. Using the BP+OSD decoder, we demonstrate good performance in the phenomenological and circuit-level noise settings, comparable to the surface code with similar distances. Finally, we conduct an analysis of the space-time overhead incurred by these codes.

\end{abstract}
\maketitle
\section{Introduction} 

Quantum computers have the potential to deliver extraordinary computational power by harnessing unique quantum phenomena such as superposition and entanglement. However, quantum systems are very susceptible to noise, and scalable quantum computers will likely require mechanisms to systematically and reliably detect and correct errors during computations. To this end, Shor introduced the first quantum error correcting code~\cite{shor_9qubit_code}, demonstrating the possibility of using redundancy to encode and protect quantum information. This work was followed by Kitaev's introduction of the surface code~\cite{Kitaev_1997}, which has become the standard for high-threshold quantum error correction (QEC) and is the basis for the majority of experimental work studying QEC such as Ref.~\cite{Krinner_2022, Zhao_2022, google_surface_code}. While its excellent performance and locality make it desirable for experimental realization, the overhead required for the surface code with the error rates of current devices represents a significant barrier to implementing it as a scalable and efficient model for quantum error correction. 

Recently, there has been an increasing amount of research investigating more efficient alternatives for QEC, most notably the class of quantum low-density parity check (qLDPC) codes, much of which is summarized in Ref.~\cite{Breuckmann_LDPC}. More general than the surface code, this class of codes is characterized by sparse parity check matrices and non-vanishing encoding rates in return for reduced geometric locality. Early constructions include the Freedman-Meyer-Luo \cite{freedman_meyer_luo}, hypergraph product \cite{Tillich_hgp}, fiber bundle \cite{fiber_bundle}, lifted product \cite{lifted_product} and balanced product \cite{Balanced_prod} codes, each progressively achieving higher encoding rates and relative distances by utilizing mathematical tools from topology and graph theory. In 2021, Panteleev and Kalachev published the first construction of \textit{asymptotically good} qLDPC codes \cite{panteleev2020}, characterized by constant encoding rate and a minimum distance linear in the number of physical qubits. Since then, two more asymptotically good constructions have emerged, namely Dinur-Hsieh-Lin-Vidick codes \cite{dinur2021locallytestablecodesconstant} and quantum Tanner codes \cite{leverrier2022quantumtannercodes}, paving the way for qLDPC codes to become a promising alternative for QEC. 

While asymptotically good constructions of qLDPC codes provide the theoretical basis for scalable quantum error correction, near-term hardware is small in size, generally consisting of tens to hundreds of qubits, and exhibits relatively high physical error rates. To test the potential of qLDPC codes for use in the near term, then, there is a need to construct explicit instances of error-correcting codes and evaluate their performance in numerical simulations. Such numerical work allows us to determine whether new codes are suitable for implementation and how they would compare to other codes, including surface codes, in such implementations. This has been the focus of recent work examining explicit constructions and their performances, including for variations of the bicycle code constructions~\cite{Bravyi_BB_codes,ye_trapped_ions,lin2023quantumtwoblockgroupalgebra,koukoulekidis2_gb},   hypergraph product codes \cite{grospellier2019numericalstudyhypergraphproduct,higgot_hgp}, quantum Tanner codes \cite{guemard2025moderatelengthliftedquantumtanner} and others~\cite{lifted_product,scruby2024highthresholdlowoverheadsingleshotdecodable,Breuckmann_londe}.

In this paper, we construct a number of explicit instances of quantum Tanner codes featuring high encoding rates and relative distances, good performance in numerical simulations in relevant noise regimes, and reduced overheads when compared to surface codes. These codes range in size from 36 to 250 qubits and exhibit encoding rates around 20\%. Moreover, the smaller instances have weight-6 stabilizer generators and low space-time overhead when compared to surface codes of similar distances. We conduct numerical simulations with phenomenological and circuit-level noise using the BP+OSD decoder \cite{lifted_product, Roffe_LDPC_Python_tools_2022} to compute high pseudo-thresholds and low logical error rates in the $p \simeq 10^{-3}$ physical error rate regime.

The remainder of this paper is structured as follows. In Section II, we review quantum stabilizer and CSS codes before describing the quantum Tanner code construction as proposed by Leverrier and Zémor~\cite{leverrier2022quantumtannercodes}. Additionally, we review decoding strategies for our numerical simulations. In Section III, we provide the details of the construction of the explicit instances of our quantum Tanner codes, including their parameters and other underlying properties. In Section IV, we present our numerical results, beginning with the error models used, before describing simulation results and pseudo-thresholds. We conclude the section with a comparison of space-time overheads. Finally, we summarize our findings and discuss future work in Section V.

\section{Background}

In this section, we provide the necessary technical background related to quantum stabilizer and CSS codes, before delving into the details of the quantum Tanner code construction from \cite{leverrier2022quantumtannercodes}, including the left-right Cayley complexes and classical codes that form the basis of these codes. We also review the belief propagation and ordered-statistics decoding approach widely employed for numerical simulations with qLDPC codes.

\subsection{Quantum Stabilizer and CSS codes }
The state space of a system of $n$ qubits is a Hilbert space $\mathcal{H} = (\mathbb{C}^2)^{\otimes n}$. A quantum code encoding $k$ logical qubits is a subspace of $\mathcal{H}$ of dimension $2^k$. The most well-known method for describing quantum codes is provided by the stabilizer formalism, a group-theoretic framework introduced in \cite{gottesman_thesis}. In this formalism, a quantum code is given by the $+1$-eigenspace of an abelian subgroup $\mathcal{S}$ of the Pauli group $\mathcal{P} = \langle X_i, Y_i, Z_i | i \in [n]\rangle $, not containing $-I$. In other words, stabilizers $s \in \mathcal{S}$ are $n$-fold tensor products of Pauli operators which commute pairwise and satisfy $s|\phi \rangle =|\phi \rangle $ for any codeword $|\phi \rangle$, i.e. codewords are ``stabilized" by the elements of $\mathcal{S}$. The set of elements of $\mathcal{P}$ which commute with $\mathcal{S}$ is called the centralizer $\mathcal{C}(\mathcal{S})$. The cosets $\mathcal{C}(\mathcal{S})\backslash \mathcal{S}$ correspond to the set of logical operators of the quantum code, of which there are $4^k$. The smallest weight of any non-trivial logical operator is called the distance of the code, $d$, and roughly captures the error-correcting capabilities of the code.

Calderbank-Shor-Steane (CSS) codes are quantum error correcting codes given by two binary linear classical codes $C_X$ and $C_Z$ such that $C_X \subset C_Z^{\perp}$ \cite{cayley_graph}, \cite{steane}. From the perspective of their respective parity check matrices $H_X$ and $H_Z$, this condition is equivalent to $H_XH_Z^{T} = 0 \mod 2$. If the rows of $H_X$ and $H_Z$ are viewed as stabilizers, where a 1 in position $i$ of a row in $H_X$ $(H_Z)$ corresponds to the operator $X_i$ $(Z_i)$ and a 0 to $I_i$, then a CSS code is a stabilizer code whose stabilizer generators each consist entirely of either $X$ or $Z$ Pauli operators (aside from identity operators). 

\subsection{Quantum Tanner Code Construction}

We describe the construction of quantum Tanner codes as in \cite{leverrier2022quantumtannercodes} by reviewing left-right Cayley complexes, classical codes, and the properties of the resulting quantum CSS codes. 

\subsubsection{Left-Right Cayley Complexes}

The left-right Cayley complex is derived from the more familiar Cayley graph. A Cayley graph $\Gamma(V, E)$ is a graph-theoretical representation of a group $G$ through a fixed set of generators $S$ not containing the identity element \cite{cayley_graph}. The vertices $g \in V$ correspond to the elements of $G$. There exists an edge between two vertices $g$ and $g'$ if and only if there exists an $s \in S$ such that $ g\cdot s = g'$, where $(\cdot)$ represents the group operation. An edge is directed unless $S$ contains $s^{-1}$ as well. If $S$ is symmetric, that is $S = S^{-1}$, then the graph is undirected.

A left-right Cayley complex is a complex constructed from a Cayley graph by introducing two sets of vertices corresponding to group elements and edges based on left and right actions of generators \cite{dinur2021locallytestablecodesconstant}. Specifically, the vertices of a left-right Cayley complex are $g_i$ where $i=0,1$ and $g_i\in G$. The edges of the left-right Cayley complex connect the two sets of vertices generating a bipartite graph. The generating set $S$ of the Cayley graph is replaced by two symmetric sets of generators $A$ and $B$ that generate edges by left and right multiplication, respectively. 
\begin{definition}[Left-Right Cayley Complex]\label{def_LRCC}
    Let $G$ be a finite group and let $A, B \subseteq G$ such that $\langle A, B \rangle = G$. Further, let $A = A^{-1}$, $B = B^{-1}$. A left-right Cayley complex $\Gamma(G, A, B)$ corresponds to a graph with 
    \begin{enumerate}
        \item the vertex set \\ $V = V_0\cup V_1 =  \{g_i|  g_i \in G, i \in \{0, 1\} \}$
        \item the edge set $E = E_A \cup E_B$, where \\ $E_A = \{(g_i, (ag)_{j} )| \; a \in A, g_i \in G, i \neq j \}$ and \\$E_B = \{(g_i, (gb)_{j})  | \; b \in B, g_i \in G, i \neq j   \} $.
    \end{enumerate}
\end{definition}

This construction yields a 2D complex with faces, as 4-cycles of the graph, of the form
\begin{equation*}
    \{ g_i, (ag)_j, (gb)_j, (agb)_i \;| \; i, j \in \{0, 1\}, i \neq j\}. 
\end{equation*}
In order to ensure that the vertices opposite each other in the faces of the complex are distinct, elements of $A$ and $B$ must not be conjugates of each other.

\begin{definition}\label{tnc def}
    Let $G$ be a finite group and let $A, B \subseteq G$ such that $\langle A, B \rangle = G$. If 
    \begin{equation*}
        \forall a \in A, \; b \in B, \; g \in G, \hspace{0.9cm} ag \neq gb
    \end{equation*}
    the left-right Cayley complex $\Gamma(G, A, B)$ is said to satisfy the \textit{total non-conjugacy condition (TNC)}.
\end{definition}

Fulfillment of the total non-conjugacy condition ensures a 2D complex structure and that each vertex has degree $\Delta_A + \Delta_B$, where $\Delta_A = |A| $ and $\Delta_B = |B|$ \cite{dinur2021locallytestablecodesconstant}. For the sake of simplicity, we will usually choose $\Delta_A = \Delta_B = \Delta$.

\subsubsection{Classical Codes}

Classical linear block codes use redundancy to encode logical information and detect and correct errors. An $[n, k]$-code, that is, a classical code encoding $k$ bits of information using $n > k$ bits, is given by a binary $k \times n$ generator matrix $G$, whose rows are a set of codewords that span the code space. Alternatively, classical codes can be defined by their parity check matrix, $H$, an $(n-k) \times n$ binary matrix whose rows represent the parity checks of the code used to detect errors. These two matrices satisfy the constraint $GH^T = 0$. In the following we will specify a code $C$ by its generator matrix which we will also call $C$.

In order to define $X$- and $Z$-type parity checks on the LRCC for the quantum Tanner code, we require a pair of binary linear classical codes $(C_A, C_B)$. The code $C_A$ encodes $\rho \Delta_A$ logical bits in $\Delta_A$ bits, for some $0<\rho < 1$, and so its generator matrix has dimensions $\rho \Delta_A \times \Delta_A$. The code $C_B$ encodes $(1-\rho)\Delta_B$ logical bits in $\Delta_B$ bits. We construct the tensor codes $C_0 = C_A \otimes C_B$ and $C_1 = C_A^{\perp} \otimes C_B^{\perp}$, where $C_i^{\perp}$ represents the dual code, which can be obtained by switching the roles of the generator and parity check matrices. We recall that $\text{dim}(C_i \otimes C_j) = \text{dim}(C_i)\text{dim}(C_j)$ and $d(C_i \otimes C_j) = d(C_i)d(C_j)$ for the minimum distances of the codes, $C_i$, and $C_j$, respectively.

\subsubsection{Quantum Tanner Codes}

In order to construct a quantum Tanner code, we select a left-right Cayley complex $\Gamma(G, A, B)$ based on a finite non-abelian group $G$. We denote the set of faces of $\Gamma(G, A, B)$ incident to a given vertex $v$ as $Q(v)$ and the complete set of faces of the LRCC as $Q$. We note that $Q(v)$ is uniquely determined by a pair $(a, b)$ for every $v \in V$. The qubits of the quantum code are placed on the faces of the LRCC, so that the number of qubits of the code is equal $|Q|$. We choose two classical codes $C_A$ and $C_B$ as described above and define $C_0$ and $C_1 $ as previously mentioned. Since the number of columns of $C_A$ is $\Delta_A$ we can label the columns of $C_A$ with elements of $A$. Having chosen a fixed association of the columns with elements of $A$ we will use the notation that the codewords of $C_A$ are binary vectors $\beta_A\in \mathbb{F}_2^{A}$. Likewise we can associate the columns of $C_B$ with elements of $B$ and given such a mapping we will say that codewords of $C_B$ are binary vectors $\beta_B \subset \mathbb{F}_2^{B}$. The chosen correspondence of bits of two classical codes to group elements yields a labeling of the columns of the tensor codes $C_0$ and $C_1$ consisting of pairs $(a, b) \in A \times B$.

In order to construct stabilizer generators on $\Gamma(G, A, B)$ using the classical codes $C_0$ and $C_1$, we can define a function
\begin{equation*}
    \phi_v: A \times B \rightarrow Q(v), \; (a, b) \mapsto \{ v, av, vb, avb \},
\end{equation*}
which maps a pair of group generators to the face in $Q(v)$ that it uniquely defines. It is easy to show that $\phi_v$ is bijective. For each basis element $\beta \in \beta_0$ of $C_0$ we can associate a set of pairs of group generators $Z(\beta) = \{(a, b) | \beta_{(a, b)} = 1 \}$ corresponding to nonzero entries of $\beta$. Each generator of the $Z$ stabilizers of the quantum Tanner code is specified by a choice of vertex $v \in V_0$ and classical codeword $\beta$ such that the $Z$ stabilizer generator has support equal to the set of faces $\phi_v(Z(\beta))$. We can characterize this stabilizer generator by a binary vector $x \in \mathbb{F}_2^{Q}$ where the $|Q|$ qubits of the classical code are labelled by faces of the LRCC. Thus for a given $Z$-stabilizer generator of the quantum code $x_{|Q(v)}$ is equal to a basis element $\beta$ of $C_0$, based on a fixed ordering of the faces, and $0$ elsewhere. 
The resulting $\dim(C_0)|V_0|$ $Z$-type stabilizer generators correspond to codewords of $C_0$ locally at each vertex. We repeat the same process for vertices $v \in V_1$ and basis elements of $C_1$ to produce $\dim(C_1)|V_1|$ $X$-type stabilizers at each vertex of the partition. 

This construction naturally yields a valid CSS code with low-weight stabilizer generators. In particular, it is shown in \cite{leverrier2022quantumtannercodes} that all stabilizer generators of opposite type commute pairwise with one another, meaning the CSS code orthogonality constraint $C_X \subset C_Z^{\perp}$ is fulfilled. A family of codes is said to exhibit the LDPC property when the number of qubits involved in every stabilizer generator and the size of the support of each stabilizer generator are bounded above by a constant that does not grow with the size of the code. Due to the fact that $|Q(v)| = \Delta^2$ for all vertices $v \in V$, all stabilizers have maximum weight $\Delta^2$. Additionally, we can count the number of parity checks that each qubit is involved in by noting that each face of the LRCC is adjacent to four vertices and each vertex corresponds to either an $X$ stabilizer generator or a $Z$ stabilizer generator. These generators arise from the parity check matrices of $C_1$ and $C_0$ respectively, which have $\rho(1-\rho)\Delta^2$ rows. So each qubit is involved in a maximum of $4 \rho(1-\rho)\Delta^2$ stabilizer generators. Defining a family of quantum Tanner codes by fixing $\Delta$ and choosing groups $G$ such that $|G| \rightarrow \infty$, it is clear that any family of quantum Tanner codes exhibits the LDPC property.

An examination of the parameters of quantum Tanner codes in terms of the properties of the LRCCs and classical codes from which they are constructed proves they are also asymptotically good. By a simple counting argument, we have $n = \Delta^2|G| /2$. Counting the $X$- and $Z$-type stabilizers yields $k \geq |V_0| \dim(C_0) + |V_1|\dim(C_1)$. Recalling the dimensions of $C_A$ and $C_B$, it is easy to show that $\dim(C_0) = \dim(C_1) = \rho(1-\rho)\Delta^2$. This, in turn, means that
\begin{equation}
    k \geq 4\rho(1-\rho)n.
\end{equation}

Current lower bounds on the distances of quantum Tanner codes which scale linearly with $n$ hinge on the expansion properties of the left-right Cayley complexes and the distance and robustness of the classical codes underlying them. In particular, LRCCs, when considered as graphs, must be Ramanujan or nearly-Ramanujan. An $r$-regular graph $G$ is said to be Ramanujan if $\lambda_1 \leq 2 \sqrt{r-1}$, where $\lambda_1$ denotes the second largest eigenvalue (in absolute terms) of the adjacency matrix of the graph $G$. The Ramanujan property represents maximal spectral expansion, which can be thought of as the ideal balance between connectivity and edge-sparsity of a graph. The classical codes and their duals in these results are assumed to have sufficiently large minimal distances and the dual tensor codes $C_0^{\perp}$ and $C_1^{\perp}$ are assumed to exhibit a property called \textit{$\kappa$-robustness}. The former condition is common in product constructions such as \cite{Tillich_hgp} and \cite{Kovalev_GB_code}. The latter appears in \cite{dinur2021locallytestablecodesconstant} and \cite{panteleev2020}, and it has been shown that $\kappa$-robustness can be achieved with high probability when $C_A$ and $C_B$ are chosen randomly \cite{dinur2022goodquantumldpccodes} \cite{kalachev2023twosidedrobustlytestablecodes}. 

By requiring $C_A, C_B, C_A^{\perp}$ and $C_B^{\perp}$ to have distance at least $\delta \Delta$ for some $\delta>0$  and $C_0^{\perp}$ and $C_1^{\perp}$ to be $\kappa$-robust, is is shown in \cite{leverrier2022decodingquantumtannercodes} that the distance of the resulting quantum Tanner code can be bounded by 
\begin{equation}
    d \geq \frac{\delta^2 \kappa^2}{256\Delta}n,
\end{equation}
a tighter bound than the original one presented in \cite{leverrier2022quantumtannercodes}. Thus, under these conditions, the parameters of these codes scale as $[[n, \Theta(n), \Theta(n)]]$, meaning they are asymptotically good qLDPC codes.

\subsection{Belief Propagation and Ordered-Statistics Decoding}

Error syndromes resulting from stabilizer measurements are passed to a decoder that aims to provide a suitable correction. In the classical setting, this translates to finding a minimum-weight estimate of the error $e$ such that 
\begin{equation}
    He = s
\end{equation}
for the parity check matrix $H$ of the code and a syndrome $s$. For classical LDPC codes, the belief propagation (BP) decoder \cite{Kschischang_bp} uses bit-wise marginal probabilities to deliver the most likely minimum-weight error. More specifically, the decoder begins by computing a marginal distribution for each bit of the error vector
\begin{equation}
    P(e_i = 1) = \sum_{j \in [n]\backslash\{i\}} P(e_1, e_2, \ldots, e_i=1, \ldots, e_n | s)
\end{equation}
given the syndrome $s$. This is the probability, given $s$, that an error has occurred on the $i^\text{th}$ bit. This marginal distribution is called the \textit{soft decision} for the bit $e_i$. The \textit{hard decision} $\hat{e}$ is the error vector obtained by setting $\hat{e}_i = 1$ if $P(e_i=1) \geq 1/2$ and $0$ otherwise. For some codes, the soft decision can be computed efficiently via factorization informed by the structure of the code's factor graph. In each iteration of the decoder, the marginal distribution $P(e_i = 1)$ is updated for each bit via the factor graph factorization and the validity of corresponding hard decision vector $\hat{e}$ is verified via $H \cdot \hat{e} = s$. This process is repeated a maximum of $n$ times, where $n$ is the length of the code. If the equation $H \cdot \hat{e} = s$ is at any point satisfied, the algorithm has converged, and $\hat{e}$ is returned as the correction. 

In the context of quantum CSS codes with uncorrelated $X$ and $Z$ errors, it first appears as if this process can be applied to the $X$- and $Z$-components of the code separately. However, this overlooks the issue of quantum degeneracy, which is rooted in the uniquely quantum phenomenon of superposition.  In the quantum setting, the goal is to return the state to the code space, meaning that operations equivalent up to a stabilizer are equally valid. When the BP decoder is applied directly to quantum codes, multiple equivalent corrections are assigned high probabilities, leading to a scenario called split belief \cite{split_belief}. The BP decoder returns the sum of these corrections, which no longer returns the state to the code space and the decoding fails. 

To circumvent the issue of degeneracy and non-convergent BP decoding, Panteleev and Kalachev \cite{lifted_product} proposed using \textit{ordered statistics decoding} (OSD) as a post-processing measure. First introduced for classical codes in Ref.~\cite{osd_fossorier}, the OSD algorithm uses submatrix inversion to deliver estimates for errors that have occurred. More specifically, this means selecting a linearly independent subset of columns of the parity check matrix $H$, denoted $[I]$ and termed the \emph{information set}. The submatrix of $H$ consisting only of these columns $H_{[I]}$ can be inverted to give a solution $e_{[I]} = H^{-1}_{[I]}\cdot s$. Because each choice of $[I]$ yields a unique solution $e_{[I]}$, the issue of degeneracy can be avoided. 

When used as a post-processing step following BP decoding on the $X$- or $Z$-component of a quantum code, the soft decision can be used to inform the choice of the information set $[I]$. The indices of a set $[n]$, which correspond to the qubits of the code, are first ordered according to the soft information from most to least likely to have been flipped, yielding an information set $[L]$. The columns of the parity check matrix $H$ are reordered according to the new ordering given by $[L]$, denoted  $H_{[L]}$. The OSD step is applied to $H_{[I]}$, where $[I]$ is the set of the first $\text{rank}(H)$ columns of $H_{[L]}$, to yield $e_{[I]} = H^{-1}_{[I]} \cdot s$. The remainder of the bits, i.e., those which are not elements of $[I]$, are set to $0$ and the bits of the solution $e = (e_{[I]}, 0)$ are returned to their original ordering. This process, called OSD-0 post-processing, can be generalized using a greedy algorithm to assign highly likely values to the remainder of the qubits not in $[I]$ based on the soft information passed down from the BP step.  This is known as higher-order OSD. One variation of this method is called the ``combination sweep'' strategy \cite{roffe_BPOSD_package} and prioritizes low weight configurations for some number $\lambda \leq n-|[I]|$ of the remaining bits.

\section{Explicit Instances of Quantum Tanner Codes}\label{explicit constructions}

\begin{table*}[ht]
\begin{center}
\begin{tabular}{|| c | c | c| c | c | c ||}
 \hline
$[[n, k, d]]$ &  Group &  $\Delta$ & Stabilizer Weights  & Encoding Rate  $(k/n)$  & Relative Distance  $(d/n)$  \\ 
\hline

[[36, 8, 3]] &$ D_4$ &  3& 6 & 0.222 & 0.083   \\ 
\hline

[[54, 11, 4]]& $D_6$ & 3 & 6  & 0.204 & 0.074  \\ 
\hline

[[72, 14, 4]]& $D_8$ & 3 & 6 & 0.194 & 0.056 \\ 
\hline

[[200, 10, 10]]&  $D_8$ &  5 &6, 8, 9, 12& 0.05 & 0.05  \\ 
\hline

[[250, 10, 15]]& $D_{10}$ &5 & 6, 8, 9, 12& 0.04  &0.06  \\ 
\hline

\end{tabular}
\end{center}
\caption{Parameters of a collection of quantum Tanner codes. The parameters of the codes are listed in the first column, where $n$, $k$, and $d$ represent the number of physical qubits, or length, the number of logical qubits, or dimension, and distance of the codes, respectively. The relative distance and encoding rate allow for a comparison between codes of different sizes across families of codes. Stabilizer weights for smaller codes are limited to 6, while they double for larger codes.}
\label{params}
\end{table*}

In this section, we provide the details for our constructions of explicit instances of quantum Tanner codes, including the groups and $\Delta$ values selected for the left-right Cayley complexes and considerations related to the choice of classical codes. We examine the parameters and stabilizer weights of the explicit instances of these codes and examine the effects of the properties of the underlying LRCCs and classical codes on the properties of the resulting quantum codes.

In order to build explicit instances of quantum Tanner codes small enough for meaningful numerical simulations and potential applications on near-term hardware, the requirements for asymptotically good parameters must be balanced with more practical considerations. Spectral expansion properties are emphasized in Ref.~\cite{leverrier2022quantumtannercodes} and necessary for a lower bound on the distance that is linear in the length of the code $n$. Because LRCCs and Cayley graphs constructed from the same group and generating sets are strongly related (see Appendix \ref{Appendix0}), it is sufficient to consider the expansion properties of Cayley graphs, which are well-studied. The first explicit Ramanujan Cayley graphs described in the literature were constructed with projective special linear (PSL) groups in Ref.~\cite{LPS_ramanujan}. These, however do  not provide small LRCCs that are suitable for constructing codes small enough for numerical simulations. Hirano \textit{et al.}~\cite{HIRANO} show that Frobenius groups, when paired with generating sets of certain sizes, also yield Ramanujan Cayley graphs. Among these, dihedral groups $D_n$, which are of order $2n$ and represent the symmetries of an $n$-gon, prove most suitable based on their slow growth in $n$ and Ramanujan properties at scale. Specifically, we select the dihedral groups $D_4, D_6, D_8, D_{10}$ and random symmetric sets of generators of size $\Delta$ to generate LRCCs. These $\Delta$ values are limited by the cardinality of the group $G$, as they must satisfy $\Delta < |G|/2$, as well as certain group-theoretic properties. Further details can be found in Appendix \ref{AppendixA}. We note that similar considerations for group choice were made in Ref.~\cite{guemard2025moderatelengthliftedquantumtanner}.

We construct explicit instances of quantum Tanner codes by combining these left-right Cayley complexes based on dihedral groups with pairs of classical codes. Finding code pairs with appropriate dimensions severely restricted the search, particularly within the realm of code families such as Reed-Muller codes. Instead, we use classical codes obtained by randomly generating a matrix $P$ of dimensions $\rho\Delta \times \Delta(1-\rho)$ and constructing the generator and parity check matrices as $G = [\mathbb{I}_{\rho\Delta} | P]$ and $H = [P^T | \mathbb{I}_{\Delta(1-\rho)}]$. We computed the minimum distances of these codes and then tested for robustness, one of the properties central to the argument for asymptotically good parameters in the resulting quantum codes. Ultimately, robustness appeared to have little to no impact on the parameters of the resulting quantum Tanner codes, while large minimal distances were necessary for high-distance quantum codes. Further details related to the construction of these explicit instances can be found in Appendix \ref{appendixB}.

A selection of the quantum Tanner codes resulting from this construction is presented in Table~\ref{params}. The first column lists the parameters of the codes, which vary in the number of physical qubits from 36 to 250. The relative distances and encoding rates have been included to facilitate a comparison across codes and with other families of qLDPC codes. Ideally, error-correcting codes exhibit both a high relative distance and a high encoding rate, meaning, respectively, that they can correct high-weight errors and encode more logical qubits with less space overhead. There is a noticeable trade-off between encoding rate and relative distance; while these codes exhibit high encoding rates around 20\%, they have slightly lower relative distances than other instances of qLDPC codes such as the bivariate bicycle codes constructed in Refs.~\cite{Bravyi_BB_codes, shaw} and lifted quantum Tanner codes \cite{guemard2025moderatelengthliftedquantumtanner}, though the latter family of codes only encodes two logical qubits. We note here that all minimum distances are given by upper bounds computed by the \texttt{BP+OSD} decoder from the \texttt{LDPC} package \cite{roffe_BPOSD_package} and verified by the \texttt{GAP} package \texttt{QDistRnd} \cite{Pryadko_qdistrnd}.

Constructing explicit instances of quantum Tanner codes also allows us to  qualitatively examine the impact of the properties of the LRCC and classical codes on the parameters of the resulting quantum codes. While the expansion properties of the underlying left-right Cayley complexes and the robustness of the classical code pairs have no clear impact on the parameters of the resulting quantum Tanner codes, noticeable trends emerge related to the $\Delta$ values and minimal distances of the classical codes. Table \ref{elements} lists the largest distances achieved by our quantum Tanner codes as a function of the $\Delta$ values and minimal distances of the classical code pairs $(C_A, C_B)$, denoted $(d_A, d_B)$, with which they were constructed. Two clear patterns emerge from this table. Firstly, we observe that the distances of the quantum Tanner codes grow with the distances of the classical codes. This is the case for many product constructions of qLDPC codes, such as hypergraph product codes \cite{Tillich_hgp}, in which the distances of the quantum codes are bounded in some way by the distances of the underlying classical codes. Secondly, and perhaps more subtly, the odd $\Delta$ values yield quantum codes with larger distances, even when the classical codes have smaller distances. It is unclear whether this is related to the choice of group or if this is a more general pattern. 

\begin{table}
    \begin{center}

        \begin{tabular}{|| c | c | c| c | c | c |c |c |c ||}
            \hline 
            \diagbox[width=\dimexpr 1.5\textwidth/17+3\tabcolsep\relax, height=1.1cm]{($d_A, d_B$) }{\\$\Delta$}& \; 3 \; & \; 4 \; & \; 5 \; & \;  6 \;  \\ \hline
             (1,1)& 1 & 1& 3& 3 \\
            \hline
            (1,2)& 1& 1& 4 & 3 \\
            \hline
             (2,1)& 3 & 2 & 7 & 5 \\
            \hline
            (2,2)&3 & 2 & 10 & 6 \\
            \hline 
            (3,1) & 4 & 3 & 8 & 6 \\
            \hline
            (3,2) & 4 & 3 & 15 & 7 \\
            \hline
            (4,1) &  - & 3 & - & 7\\
            \hline
            (4,2) & - & 4 & - & 10\\

            \hline
            
        \end{tabular}
    \end{center}
    \caption{The maximal distances of the quantum codes constructed are listed as a function of the $\Delta$ value used in the construction of the LRCCs (top row) and the minimal distances of the classical code pairs, denoted $(d_A, d_B)$ (left column). Classical codes with higher distances appear to yield quantum Tanner codes with higher distances. Additionally, we see that these distances grow more quickly with odd $\Delta$ values. Codes constructed with $\Delta = 5$ and classical pairs with minimal distances $(d_A, d_B) \in \{(4,1), (4,2) \}$ did not yield valid CSS codes.}
    \label{elements}
\end{table}

\section{Numerical Results}

In this section, we introduce the error models used for numerical simulations and then investigate the pseudo-thresholds under these error models of the codes presented in the previous section. We discuss simulation results in the phenomenological and circuit-level noise settings and compare these with surface codes of similar distance. Finally, we analyze the space-time overheads of the codes to facilitate a comparison with their surface code counterparts as well as other qLDPC codes. All code and data can be found at \url{https://github.com/RebKatRad/qTanner.git}.

\subsection{Error Models}

In order to evaluate the error-correcting capabilities of the explicit instances of quantum Tanner codes in the context of fault-tolerant quantum computing, we perform numerical memory experiments which capture their performance in systems affected by different types of noise, including noise in the syndrome extraction process itself. We go beyond the standard code capacity setting, in which Pauli noise is applied to data qubits and perfect syndrome extraction is assumed, by introducing phenomenological and circuit level noise. In the simulations, data qubits are initialized in the $X$ ($Z$) basis and $N \in \mathbb{N}$ rounds of syndrome extraction are performed, wherein all stabilizers of the code are measured using ancilla qubits to yield a syndrome and $N$ is chosen depending on the type of experiment. Subsequently, all data qubits are measured in the $X$ ($Z$) basis in a final destructive round of measurement, simulating logical measurement, from which the stabilizer readout can be reconstructed. 

The phenomenological noise model is a heavily simplified model used to study the effect of both data qubit and measurement errors on the error correction capabilities of a code. In each round of syndrome extraction, each data qubit is independently subject to a $Z$ ($X$)-type error with probability $p$. In the first $N$ rounds of syndrome extraction, each stabilizer measurement independently reports an incorrect result with the same probability $p$. The final round of measurement is noiseless. Data-qubit errors, which have accumulated across the $N$ rounds, and measurement errors have equivalent effects and potentially result in a logical error, which is detected with perfect precision due to the lack of measurement noise in the final round of measurement. 

The circuit-level noise model is more detailed and aims to reproduce the noise exhibited by certain hardware models. Each stabilizer measurement is performed using an ancilla qubit together with two-qubit Clifford gates performed between each data qubit in the support of the stabilizer and the ancilla, and finally the ancilla qubit is measured.  In our simulations, we apply i.i.d.\ single- and two-qubit Pauli errors with probability $p$ following non-trivial single- and two-qubit Clifford gates, respectively, and after reset operations. Data qubits idling during time steps in which they are not being measured as well as idling ancilla qubits exhibit errors with probability $p/10$, simulating the common behavior that idling errors are far less probable than other types~\cite{Conrad_dodeco}. 

Simulating circuit-level noise requires generating a circuit and inserting errors of the appropriate intensity at right locations. To this end, we use the \texttt{LDPC} package (version 1)\footnote[1]{Note that we used the \texttt{css$\_$code$\_$memory$\_$circuit} functionality recently introduced by  Higgott that includes idling errors and measures both types of stabilizers unlike previous versions.} \cite{Roffe_LDPC_Python_tools_2022}  to generate generic circuits, using \texttt{Stim} \cite{gidney2021stim}, for the $X$ and $Z$ components of our CSS codes and specify the probabilities for each type of error. We note that these syndrome extraction circuits have not been optimized for quantum Tanner codes, a step we leave for future work.

\subsection{Simulation Results and Pseudo-Thresholds}

Simulations provide insight into the performance of explicit codes at noise levels realistic for near-term devices and facilitate a comparison with other qLDPC codes, including the surface code. We use the \texttt{BP+OSD} decoder \cite{roffe_BPOSD_package} to compute the logical error rates at a variety of physical error rates in the presence of both measurement and circuit-level noise. Logical error rates $L_X$ and $L_Z$ are computed separately for the $X$ and $Z$ components of the code, respectively, and the overall logical error rate is approximated by 
\begin{equation}\label{logical error rate def}
    p_L = (L_X + L_Z - L_XL_Z)/N.
\end{equation}
For each distance-$d$ code, we perform $N=d$ rounds of syndrome extraction before measuring the data qubits in the final round to determine the presence of a logical error. In our simulations, we use the \texttt{LDPC} package \cite{Roffe_LDPC_Python_tools_2022} to implement the min-sum variation of BP. We found that the performance of the decoder was very sensitive to the min-sum scaling factor $\alpha$, with some values degrading the quality of the results significantly, and we used $\alpha = 0.625$ as suggested in Ref.~\cite{lifted_product}. For the post-processing step, we use OSD-9 with the combination strategy, which outperformed all other OSD-$\lambda$ for $\lambda \in \{1, \ldots, 10\}$.

\begin{figure*}[ht!]
\centering

\includegraphics[width = 0.97\linewidth]{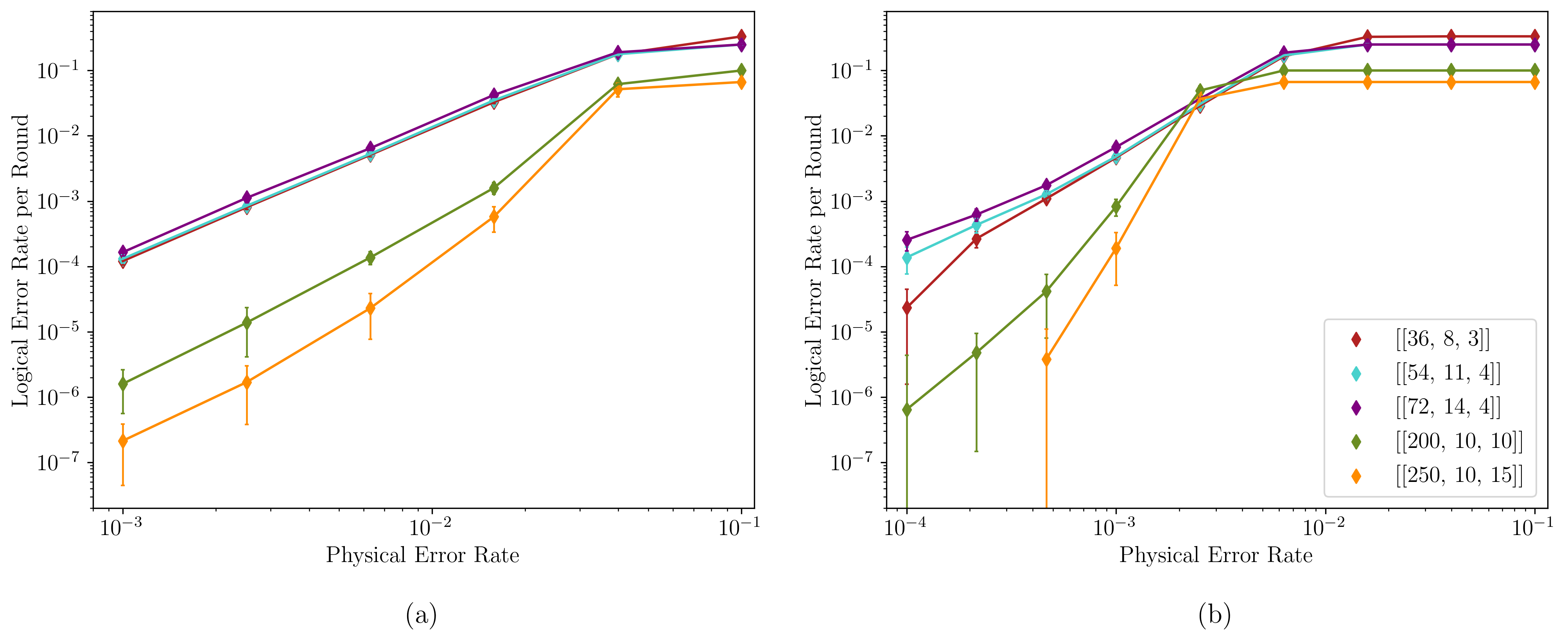}

\caption{Performance of our quantum Tanner codes from Table \ref{params} in simulations with (a) phenomenological noise, and (b) circuit-level noise. Logical error rate is plotted as a function of physical error rate for $N=d$ rounds of syndrome extraction, for code distance $d$. Error bars were computed for 95\% binomial proportion confidence intervals via $p_L = \hat{p_L} \pm \frac{1.645}{\sqrt{\eta}} \sqrt{\frac{\eta_s\eta_f}{\eta^2}}$, where $\eta$ represents the total number of shots and $\eta_s$ and $\eta_f$ the number of successes and failures, respectively. For the combined $p_L$ as computed in Eq.~\eqref{logical error rate def}, $\eta = \eta_x+\eta_z$. Here, $\eta_x, \eta_z \in [10^5, 10^9]$, depending on the number of errors encountered. Note the distinct ranges of physical error rates in plots (a) and (b).} 
\label{simulations}
\end{figure*}

In Figure \ref{simulations}(a) we see that, with measurement noise, all codes exhibit a logical error rate in the regime of $p_L$ ranging from $10^{-7}$ to $10^{-4}$ at a physical error rate of $p = 10^{-3}$, a level considered achievable by small-scale near-term hardware. Under circuit-level noise, Figure \ref{simulations}(b), logical error rates range from $10^{-4}$ to $10^{-2}$ at a physical error rate of $p = 10^{-3}$.

\begin{figure}[ht]
    \centering
    \includegraphics[width=0.99\linewidth]{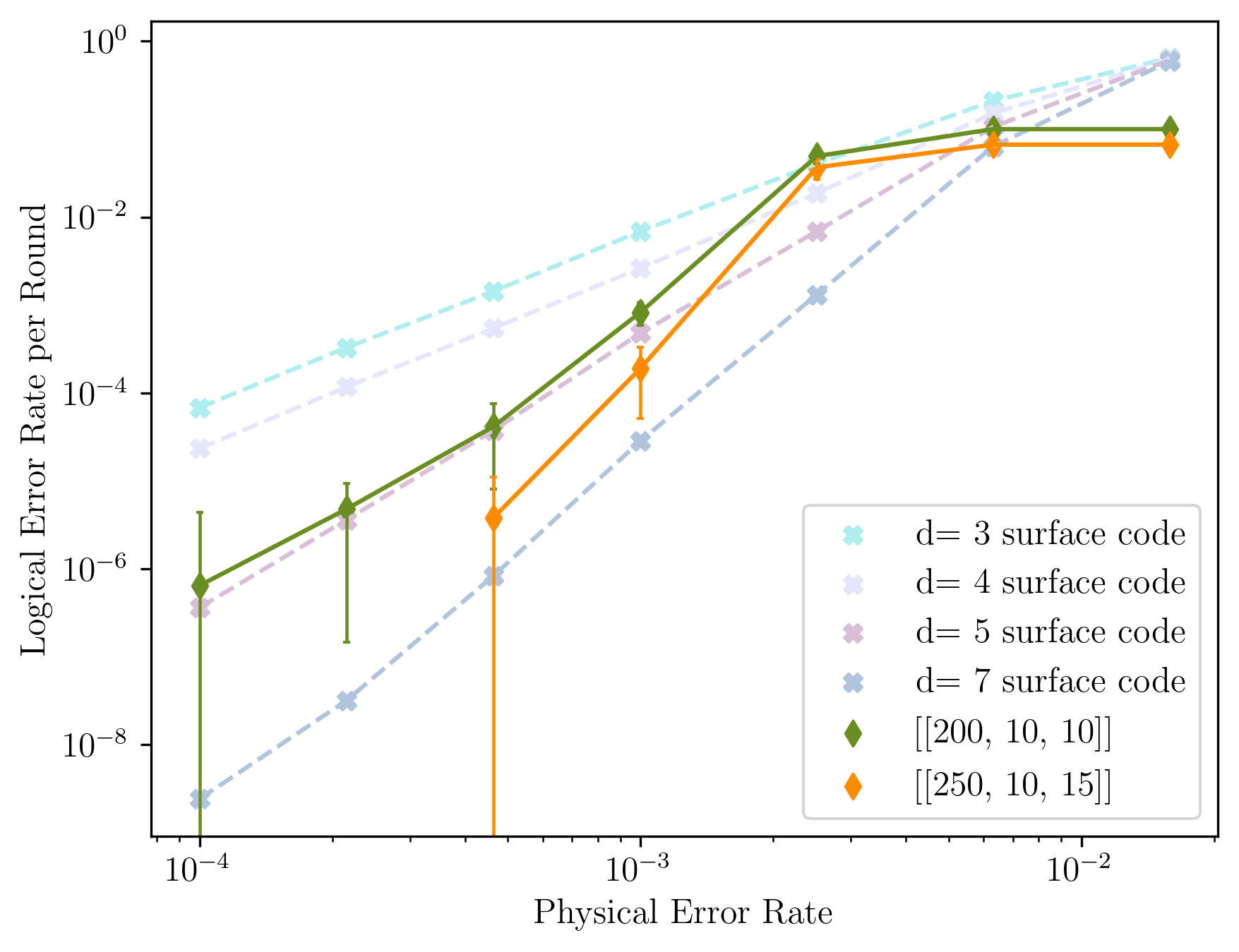}
    \caption{A comparison of the performances of quantum Tanner codes with 10 logical qubits and $k = 10$ copies of various surface codes of comparable distance. The logical error rates of the surface codes were computed via $p'_L = 1- (1-p_L)^k$ for $k  =10$, where $p_L$ is the logical error rate of a single surface code encoding one logical qubit.}
    \label{sc comparison}
\end{figure}

Figure \ref{sc comparison} shows a comparison of the performances of the quantum Tanner codes and surfaces codes of comparable distance. In order to facilitate a comparison with the $[[200, 10, 10]]$ and $[[250, 10, 15]]$ codes, the logical error rates of $k = 10$ copies of the surface codes were computed via
\begin{equation}
    p'_L = 1- (1-p_L)^k,
\end{equation}
where $p_L$ is the logical error rate of a single copy of the surface code encoding only a single logical qubit.

Pseudo-thresholds can be used to quantify the performance of the quantum Tanner codes in the presence of phenomenological and circuit level noise. The pseudo-threshold is the point at which the error rate of the logical qubits is equal to that of the physical qubits. Below this point, the system experiences sufficient suppression of logical errors in the presence of physical noise. In order to facilitate a comparison between physical and logical error rates, we examine the logical error rate per logical qubit. In the phenomenological case, we compute the pseudo-threshold as the physical error rate $p$ satisfying $p_L(p) = kp$, where $p_L(p)$ is the per-round logical error rate and $k$ represents the number of encoded qubits. The pseudo-thresholds of our quantum Tanner codes are reported in Table \ref{pseudo_thresholds}. Under phenomenological noise, these vary in the range $1.3\%$ to $6.3\%$. 
\bgroup

\def\arraystretch{1.05}
\begin{table}[!htb]
    \begin{center}
        \begin{tabular}{||c|c|c||}
            \hline
             $[[n, k, d]]$ &  \Centerstack{Pseudo-threshold \\ (phenomenological)} &  \Centerstack{Pseudo-threshold \\(circuit-level)}\tabularnewline[8pt]
             \hline 
             [[36, 8, 3]] &  0.0634 & 0.0038 \\
            \hline

            [[54, 11, 4]] & 0.0382 & 0.0056 \\
            \hline
            
            [[72, 14, 4]] &  0.0300& 0.0036 \\
            \hline

            [[200, 10, 10]] & 0.0198 & 0.0059 \\
            \hline
            
            [[250, 10, 15]] & 0.0133 & 0.0040 \\
            \hline
        \end{tabular}
    \end{center}
    \caption{Pseudo-thresholds represent the break even points where $p_L(p) = kp$ (phenomenological) or $p_L(p) = Tkp/10$ (circuit-level) for a physical error rate $p$, where $p_L$, $k$, and $T$ represent the logical error rate, number of logical qubits, and depth of a single round of syndrome extraction, respectively. Pseudo-thresholds are listed for the quantum Tanner codes for both the phenomenological and circuit-level noise settings. }
    \label{pseudo_thresholds}
\end{table}
\egroup

In the case of circuit-level noise, we consider the break even point $p_L(p) = Tkp/10$, where $T$ is the number of time steps, or depth, of a single round of syndrome extraction. We divide by 10 due to the fact that we consider idling errors with intensity $p/10$ in this error model. As can be seen in Table \ref{pseudo_thresholds}, the pseudo-thresholds of quantum Tanner codes in the circuit level model are lower than those in the phenomenological model by approximately one order of magnitude. Both sets of pseudo-thresholds remain competitive with those displayed by other qLDPC codes. These results suggest that quantum Tanner codes are able to suppress errors effectively on near-term devices consisting of around 200 physical qubits experiencing physical error rates on the order of $10^{-3}$. 

\subsection{Overhead Comparison}

We examine the space-time overheads of the explicit quantum Tanner codes generated here and compare these estimates to those incurred by surface codes which yield similar logical error rates under both phenomenological and circuit level noise. The space overhead is computed as the total number of physical qubits $n + n_{anc}$, where $n_{anc}$ represents the number of ancilla qubits required for the syndrome extraction circuit. In our case, this coincides with the number of stabilizers of the code, but can potentially be reduced in certain architectures, such as trapped ions~\cite{ye_trapped_ions}. The time overhead is taken as the  product of the depth of the syndrome extraction circuit and the number of syndrome extraction rounds. Let $d_x$ represent the maximal weight of any row or column in the $X$-parity check matrix $H_X$ and $d_z$ the maximal weight of any row or column in $H_Z$. Then the depth of a single round of syndrome extraction, in which we measure both sets of stabilizers, is bounded from above by $d_x + d_z$. As such, we can estimate the time overhead as the product $N(d_x + d_z)$, where we do $N = d$ rounds of syndrome extraction for a distance-$d$ code. Overall, we have
\begin{equation}
    O_{ST} = (n + n_{anc})(d_x + d_z)d.
\end{equation}

\bgroup
\def\arraystretch{1.3}
    \begin{table*}[ht]
    \begin{center}
        \begin{tabular}{|| c | c | c | c||}
            \hline 
             Code & Space-time overhead per logical qubit &   \shortstack{$p_L/k$ at $p = 10^{-3}$ \\ (phenomenological)} &  \shortstack{$p_L/k$ at $p = 10^{-3}$ \\ (circuit-level)}\\
            \hline
             [[36, 8, 3]] & 612  & $1.71\times 10^{-5}  \; (\pm 6 \times 10^{-7})$& $6.52\times 10^{-4}  \; (\pm 6 \times 10^{-6})$\\
            \hline
             [[54, 11, 4]] & 891 &$4.1\times 10^{-6}  \; (\pm 1 \times 10^{-7})$ & $2.38\times 10^{-4}  \; (\pm 3 \times 10^{-6})$\\
            \hline
             [[72, 14, 4]] & 933 & $1.12\times 10^{-5}  \; (\pm 1 \times 10^{-7})$ & $4.83\times 10^{-4}  \; (\pm 3 \times 10^{-6})$\\
            \hline
             [[200, 10, 10]] & 16\,464 & $1.00\times 10^{-7}  \; (\pm 8 \times 10^{-9})$ & $8.2\times 10^{-5}  \; (\pm 2 \times 10^{-6})$\\
            \hline
            [[250, 10, 15]] & 30\,870  & $5.3\times 10^{-8}  \; (\pm 5 \times 10^{-9})$ & $2.44\times 10^{-5}  \; (\pm 9 \times 10^{-7})$\\
            \hline
            \hline
            $d=3$ s.c. & 600 &  $3.0 \times 10^{-5}$& $6.9 \times 10^{-4}$ \\
            \hline
            $d = 4 $ s.c. & 1568 &  $5.0 \times 10^{-6}$& $2.6 \times 10^{-4}$\\
            \hline
            $d = 5$ s.c. & 3240 &  $8.0 \times 10^{-7}$ &  $4.8 \times 10^{-5}$ \\
            \hline
           
            $d = 7$ s.c. & 9464 &  $2.9 \times 10^{-9}$ & $2.9 \times 10^{-6}$\\
            \hline
            $d = 9$ s.c. & 20\,808 & $2 \times 10^{-10}$ & $1 \times 10^{-6}$ \\
            \hline
          
            $d = 15$ s.c. & 100\,920 & $2.3 \times 10^{-18*}$
\footnote[0]{*This value is an approximation based on a fitting of logical error rates computed at higher physical error rates.} &$2.2 \times 10^{-10}$\\
            \hline

        \end{tabular}
    \end{center}
    \caption{A comparison of the space-time overheads of quantum Tanner and surface codes of comparable distance and logical error rate per logical qubit. The surface codes have parameters $[[L^2, 1, L]]$. The value $p_L/k$ was computed using the BP+OSD decoder with $d$ rounds of syndrome extraction in the presence of phenomenological and circuit-level noise for both sets of codes.}
    \label{overheads}
\end{table*}
\egroup

Table \ref{overheads} summarizes the overhead findings and offers a comparison with surface codes of similar distance and performance. Quantum Tanner codes and distance $d$-surface codes are listed in the first column and space-time overheads per logical qubit are listed for both sets of codes in the second column. For smaller codes, quantum Tanner codes require up to 50\% less space-time overhead per logical qubit while achieving a level of error suppression comparable with their distance-3 and -4 surface code counterparts. At this scale, quantum Tanner codes represent a more efficient and equally effective alternative to the surface code. 

On the other hand, surface codes of larger distances clearly outperform the 200- and 250-qubit quantum Tanner codes in both the phenomenological and circuit-level noise settings. However, this is achieved with up to triple the space-time overhead, despite the high stabilizer weights of the $[[200, 10, 10]]$ and $[[250, 10, 15]]$ codes, making the quantum Tanner codes more efficient but less effective than surface codes in this regime. We note that these overheads may further be reduced by optimizing the syndrome extraction circuits and exploring single-shot error correction, a property which quantum Tanner codes have been proven to exhibit under adversarial noise \cite{Gu_single_shot}. 

\section{Conclusion and Future Work}

In this work, we have constructed a number of explicit instances of quantum Tanner codes, a family shown to have asymptotically good parameters. The codes were constructed using dihedral groups and random pairs of classical codes and exhibit high encoding rates and relative distances. A numerical analysis conducted with the BP+OSD decoding revealed relatively high pseudo-thresholds, both in the phenomenological and circuit level noise models, and good suppression of logical errors in the $p = 10^{-3}$ physical error rate regime. Small codes exhibited particularly low overheads when compared to surface codes, due in part to their low stabilizer weights. The overall performance of these codes suggests that they are well-suited to experimental implementation on near-term hardware with around 200 qubits and physical error rates around $p = 10^{-3}$. In particular, this includes trapped-ion-based hardware, which features high connectivity and low idling error rates and therefore lends itself well to the realization of quantum Tanner codes \cite{ye_trapped_ions}. 

Further research is needed to identify a broader class of codes with improved properties, and to explore the use of these codes in an end-to-end fault tolerant architecture for quantum computing. In order to find codes with better parameters and lower stabilizer weights, the search could be expanded to quantum Tanner codes constructed with other Frobenius groups and different classes of classical codes, or based on non-left-right-Cayley complexes, such as those in \cite{olai_generalizing}. 

In order to further reduce the overhead incurred by these codes, it would be of interest to explore optimizing syndrome extraction circuits and developing a scheme for single-shot decoding. Improving syndrome extraction circuits has the potential to reduce idling time and the propagation of errors and improve the time overhead required by the error-correcting codes by lowering the depth of the circuits. A scheme of particular interest can be found in Ref.~\cite{Tremblay_2022}. The time overhead incurred by quantum Tanner codes can also be reduced by taking advantage of their single shot property, which was proven for this code family in Ref.~\cite{Gu_single_shot}. A potential approach to this problem would be to develop an implementation of the decoder proposed by Leverrier and Zémor for quantum Tanner codes in Refs.~\cite{leverrier2022decodingquantumtannercodes} or \cite{leverrier2022efficientdecodingconstantfraction}. Finally, it would be of interest to explore logical gates on quantum Tanner codes, which could be accomplished by examining their automorphism groups as proposed in Ref.~\cite{sayginel2025faulttolerantlogicalcliffordgates}.

\section{Acknowledgments}
We would like to thank Oscar Higgott for many helpful discussions and support in the numerical work. RKR is grateful to Tom Scruby and Timo Hillmann for guidance in earlier stages of the project. This work is supported by the Australian Research Council via the Centre of Excellence in Engineered Quantum Systems (EQUS) project number CE170100009, and by the ARO through the QCISS program W911NF-21-1-0007 and IARPA ELQ program W911NF-23-2-0223. RKR is supported by the Sydney Quantum Academy.

\newpage

\bibliography{bibliography}

@article{cayley_graph,
 ISSN = {00029327, 10806377},
 URL = {http://www.jstor.org/stable/2369306},
 author = {Arthur Cayley},
 journal = {American Journal of Mathematics},
 number = {2},
 pages = {174--176},
 publisher = {Johns Hopkins University Press},
 title = {Desiderata and Suggestions: No. 2. The Theory of Groups: Graphical Representation},
 urldate = {2024-09-28},
 volume = {1},
 year = {1878}
}

@misc{dinur2021locallytestablecodesconstant,
      title={Locally Testable Codes with constant rate, distance, and locality}, 
      author={Irit Dinur and Shai Evra and Ron Livne and Alexander Lubotzky and Shahar Mozes},
      year={2021},
      eprint={2111.04808},
      archivePrefix={arXiv},
      primaryClass={cs.IT},
      url={https://arxiv.org/abs/2111.04808}, 
}

@article{LPS_ramanujan,
author = {Lubotzky, Alexander and Phillips, R. and Sarnak, P.},
year = {1988},
month = {09},
pages = {261-277},
title = {Ramanujan Graphs},
volume = {8},
journal = {Combinatorica},
doi = {10.1007/BF02126799}
}

@article{HIRANO,
   title={RAMANUJAN CAYLEY GRAPHS OF FROBENIUS GROUPS},
   volume={94},
   ISSN={1755-1633},
   url={http://dx.doi.org/10.1017/S0004972716000587},
   DOI={10.1017/s0004972716000587},
   number={3},
   journal={Bulletin of the Australian Mathematical Society},
   publisher={Cambridge University Press (CUP)},
   author={Hirano, Miki and Katata, Kohei and Yamasaki, Yoshinori},
   year={2016},
   eprint = {1503.04075},
   archivePrefix = {arXiv},
   month=sep, pages={373–383} }

@INPROCEEDINGS{leverrier2022quantumtannercodes,
  author={Leverrier, Anthony and Zémor, Gilles},
  booktitle={2022 IEEE 63rd Annual Symposium on Foundations of Computer Science (FOCS)}, 
  title={Quantum {T}anner codes}, 
  year={2022},
  eprint={2202.13641},
  archivePrefix={arXiv},
  volume={},
  number={},
  pages={872-883},
  doi={10.1109/FOCS54457.2022.00117}}

@manual{GAP4,
    organization = {The GAP~Group},
    title        = {GAP -- Groups, Algorithms, and Programming,
                    Version 4.13.1},
    year         = {2024},
    url          = {https://www.gap-system.org},
    }

@phdthesis{gottesman_thesis,
    author = {Daniel Gottesman},
    title = {Stabilizer Codes and Quantum Error Correction},
    school = {California Institute of Technology} ,
    year = {1997},
    eprint={quant-ph/9705052},
    archivePrefix={arXiv},
}

@article{Breuckmann_LDPC,
  title = {Quantum Low-Density Parity-Check Codes},
  author = {Breuckmann, Nikolas P. and Eberhardt, Jens Niklas},
  journal = {PRX Quantum},
  eprint = {2103.06309},
  archivePrefix={arXiv},
  volume = {2},
  issue = {4},
  pages = {040101},
  numpages = {19},
  year = {2021},
  month = {Oct},
  publisher = {American Physical Society},
  doi = {10.1103/PRXQuantum.2.040101},
  url = {https://link.aps.org/doi/10.1103/PRXQuantum.2.040101}
}

@misc{olai_generalizing,
      title={Generalizing Quantum {T}anner Codes}, 
      author={Olai A. Mostad and Eirik Rosnes and Hsuan-Yin Lin},
      year={2024},
      eprint={2405.07980},
      archivePrefix={arXiv},
      primaryClass={cs.IT},
      url={https://arxiv.org/abs/2405.07980}, 
}

@misc{panteleev2020,
      title={Asymptotically Good Quantum and Locally Testable Classical {LDPC} Codes}, 
      author={Pavel Panteleev and Gleb Kalachev},
      year={2022},
      eprint={2111.03654},
      archivePrefix={arXiv},
      primaryClass={cs.IT},
      url={https://arxiv.org/abs/2111.03654}, 
}

@article{Tillich_hgp,
   title={Quantum {LDPC} Codes With Positive Rate and Minimum Distance Proportional to the Square Root of the Blocklength},
   volume={60},
   ISSN={1557-9654},
   eprint = {0903.0566},
   archivePrefix={arXiv},
   url={http://dx.doi.org/10.1109/TIT.2013.2292061},
   DOI={10.1109/tit.2013.2292061},
   number={2},
   journal={IEEE Transactions on Information Theory},
   publisher={Institute of Electrical and Electronics Engineers (IEEE)},
   author={Tillich, Jean-Pierre and Zemor, Gilles},
   year={2014},
   month=feb, pages={1193–1202} }

@inproceedings{fiber_bundle, 
    series={STOC ’21},
   title={Fiber bundle codes: breaking the
            n
            1/2
            polylog(n) barrier for Quantum {LDPC} codes},
   url={http://dx.doi.org/10.1145/3406325.3451005},
   DOI={10.1145/3406325.3451005},
   eprint = {2009.03921},
   archivePrefix={arXiv},
   booktitle={Proceedings of the 53rd Annual ACM SIGACT Symposium on Theory of Computing},
   publisher={ACM},
   author={Hastings, Matthew B. and Haah, Jeongwan and O’Donnell, Ryan},
   year={2021},
   month=jun, pages={1276–1288},
   collection={STOC ’21} }

@article{lifted_product,
   title={Degenerate Quantum {LDPC} Codes With Good Finite Length Performance},
   volume={5},
   ISSN={2521-327X},
   url={http://dx.doi.org/10.22331/q-2021-11-22-585},
   DOI={10.22331/q-2021-11-22-585},
   journal={Quantum},
   eprint = {1904.02703},
   archivePrefix={arXiv},
   publisher={Verein zur Forderung des Open Access Publizierens in den Quantenwissenschaften},
   author={Panteleev, Pavel and Kalachev, Gleb},
   year={2021},
   month=nov, pages={585} }

@article{Balanced_prod,
   title={Balanced Product Quantum Codes},
   volume={67},
   ISSN={1557-9654},
   url={http://dx.doi.org/10.1109/TIT.2021.3097347},
   DOI={10.1109/tit.2021.3097347},
   number={10},
   eprint = {2012.09271},
   archivePrefix={arXiv},
   journal={IEEE Transactions on Information Theory},
   publisher={Institute of Electrical and Electronics Engineers (IEEE)},
   author={Breuckmann, Nikolas P. and Eberhardt, Jens N.},
   year={2021},
   month=oct, pages={6653–6674} }

@article{Bravyi_BB_codes,
   title={High-threshold and low-overhead fault-tolerant quantum memory},
   volume={627},
   ISSN={1476-4687},
   url={http://dx.doi.org/10.1038/s41586-024-07107-7},
   DOI={10.1038/s41586-024-07107-7},
   number={8005},
   journal={Nature},
   publisher={Springer Science and Business Media LLC},
   eprint = {2308.07915},
   archivePrefix={arXiv},
   author={Bravyi, Sergey and Cross, Andrew W. and Gambetta, Jay M. and Maslov, Dmitri and Rall, Patrick and Yoder, Theodore J.},
   year={2024},
   month=mar, pages={778–782} }

@article{Pryadko_qdistrnd,
   title={QDistRnd: A {GAP} package for computing the distance of
quantum error-correcting codes},
   volume={7},
   ISSN={2475-9066},
   url={http://dx.doi.org/10.21105/joss.04120},
   DOI={10.21105/joss.04120},
   number={71},
   journal={Journal of Open Source Software},
   publisher={The Open Journal},
   author={Pryadko, Leonid P. and Shabashov, Vadim A. and Kozin, Valerii K.},
   year={2022},
   eprint = {2308.15140},
   archivePrefix = {arXiv},
   month=mar, pages={4120} }

@misc{guemard2025moderatelengthliftedquantumtanner,
      title={Moderate-length lifted quantum {T}anner codes}, 
      author={Virgile Guemard and Gilles Zémor},
      year={2025},
      eprint={2502.20297},
      archivePrefix={arXiv},
      primaryClass={quant-ph},
      url={https://arxiv.org/abs/2502.20297}, 
}

@article{shaw,
  title = {Lowering Connectivity Requirements for Bivariate Bicycle Codes Using Morphing Circuits},
  author = {Shaw, Mackenzie H. and Terhal, Barbara M.},
  journal = {Phys. Rev. Lett.},
  volume = {134},
  issue = {9},
  pages = {090602},
  numpages = {5},
  year = {2025},
  month = {Mar},
  eprint = {2407.16336},
  archivePrefix = {arXiv},
  publisher = {American Physical Society},
  doi = {10.1103/PhysRevLett.134.090602},
  url = {https://link.aps.org/doi/10.1103/PhysRevLett.134.090602}
}

@article{roffe_BPOSD_package,
   title={Decoding across the quantum low-density parity-check code landscape},
   volume={2},
   ISSN={2643-1564},
   url={http://dx.doi.org/10.1103/PhysRevResearch.2.043423},
   DOI={10.1103/physrevresearch.2.043423},
   number={4},
   journal={Physical Review Research},
   publisher={American Physical Society (APS)},
   author={Roffe, Joschka and White, David R. and Burton, Simon and Campbell, Earl},
   eprint = {2005.07016},
   archivePrefix = {arXiv},
   year={2020},
   month={Dec}
}

@article{Gu_single_shot,
   title={Single-Shot Decoding of Good Quantum {LDPC} Codes},
   volume={405},
   ISSN={1432-0916},
   url={http://dx.doi.org/10.1007/s00220-024-04951-6},
   DOI={10.1007/s00220-024-04951-6},
   number={3},
   journal={Communications in Mathematical Physics},
   publisher={Springer Science and Business Media LLC},
   author={Gu, Shouzhen and Tang, Eugene and Caha, Libor and Choe, Shin Ho and He, Zhiyang and Kubica, Aleksander},
   year={2024},
   eprint = {2306.12470},
   archivePrefix = {arXiv},
   month=mar }

@article{Kovalev_GB_code,
  title = {Quantum Kronecker sum-product low-density parity-check codes with finite rate},
  author = {Kovalev, Alexey A. and Pryadko, Leonid P.},
  journal = {Phys. Rev. A},
  volume = {88},
  issue = {1},
  pages = {012311},
  numpages = {13},
  year = {2013},
  month = {Jul},
  publisher = {American Physical Society},
  doi = {10.1103/PhysRevA.88.012311},
  url = {https://link.aps.org/doi/10.1103/PhysRevA.88.012311}
}

@article{Conrad_dodeco,
   title={The small stellated dodecahedron code and friends},
   volume={376},
   ISSN={1471-2962},
   url={http://dx.doi.org/10.1098/rsta.2017.0323},
   DOI={10.1098/rsta.2017.0323},
   number={2123},
   journal={Philosophical Transactions of the Royal Society A: Mathematical, Physical and Engineering Sciences},
   publisher={The Royal Society},
   author={Conrad, J. and Chamberland, C. and Breuckmann, N. P. and Terhal, B. M.},
   year={2018},
   eprint = {1712.07666},
   archivePrefix  = {arXiv},
   month=may, pages={20170323} }

@misc{Roffe_LDPC_Python_tools_2022,
author = {Roffe, Joschka},
title = {{LDPC: Python tools for low density parity check codes}},
url = {https://pypi.org/project/ldpc/},
year = {2022}
}

@article{gidney2021stim,
  doi = {10.22331/q-2021-07-06-497},
  url = {https://doi.org/10.22331/q-2021-07-06-497},
  title = {Stim: a fast stabilizer circuit simulator},
  author = {Gidney, Craig},
  journal = {{Quantum}},
  issn = {2521-327X},
  publisher = {{Verein zur F{\"{o}}rderung des Open Access Publizierens
                in den Quantenwissenschaften}},
  volume = {5},
  pages = {497},
  month = {jul},
  year = {2021},
  eprint = {2103.02202},
  archivePrefix = {arXiv}
}

@misc{ye_trapped_ions,
      title={Quantum error correction for long chains of trapped ions}, 
      author={Min Ye and Nicolas Delfosse},
      year={2025},
      eprint={2503.22071},
      archivePrefix={arXiv},
      primaryClass={quant-ph},
      url={https://arxiv.org/abs/2503.22071}, 
}

@ARTICLE{Kschischang_bp,
  author={Kschischang, F.R. and Frey, B.J. and Loeliger, H.-A.},
  journal={IEEE Transactions on Information Theory}, 
  title={Factor graphs and the sum-product algorithm}, 
  year={2001},
  volume={47},
  number={2},
  pages={498-519},
  doi={10.1109/18.910572}}

@article{split_belief,
  doi = {10.22331/q-2018-10-19-102},
  url = {https://doi.org/10.22331/q-2018-10-19-102},
  title = {Multi-path {S}ummation for {D}ecoding 2{D} {T}opological {C}odes},
  author = {Criger, Ben and Ashraf, Imran},
  journal = {{Quantum}},
  issn = {2521-327X},
  publisher = {{Verein zur F{\"{o}}rderung des Open Access Publizierens in den Quantenwissenschaften}},
  volume = {2},
  pages = {102},
  eprint = {1709.02154},
  archivePrefix = {arXiv},
  month = oct,
  year = {2018}
}

@ARTICLE{osd_fossorier,
  author={Fossorier, M.P.C. and Shu Lin},
  journal={IEEE Transactions on Information Theory}, 
  title={Soft-decision decoding of linear block codes based on ordered statistics}, 
  year={1995},
  volume={41},
  number={5},
  pages={1379-1396},
  doi={10.1109/18.412683}}

@article{steane, volume={452},
    title = {Multiple Particle Interference and Quantum Error Correction},
   ISSN={1471-2946},
   url={http://dx.doi.org/10.1098/rspa.1996.0136},
   DOI={10.1098/rspa.1996.0136},
   number={1954},
   journal={Proceedings of the Royal Society of London. Series A: Mathematical, Physical and Engineering Sciences},
   publisher={The Royal Society},
   year={1996},
   author = {Steane, A.},
   eprint={9601029},
   archivePrefix={arXiv},
   month=nov, pages={2551–2577} }

@article{Tremblay_2022,
   title={Constant-Overhead Quantum Error Correction with Thin Planar Connectivity},
   volume={129},
   ISSN={1079-7114},
   url={http://dx.doi.org/10.1103/PhysRevLett.129.050504},
   DOI={10.1103/physrevlett.129.050504},
   number={5},
   journal={Physical Review Letters},
   publisher={American Physical Society (APS)},
   author={Tremblay, Maxime A. and Delfosse, Nicolas and Beverland, Michael E.},
   year={2022},
   eprint = {2109.14609},
   archivePrefix = {arXiv},
   month=jul }

@article{leverrier2022decodingquantumtannercodes,
author = {Leverrier, Anthony and Z\'{e}mor, Gilles},
title = {Decoding Quantum {T}anner Codes},
year = {2023},
issue_date = {Aug. 2023},
publisher = {IEEE Press},
volume = {69},
number = {8},
issn = {0018-9448},
url = {https://doi.org/10.1109/TIT.2023.3267945},
doi = {10.1109/TIT.2023.3267945},
journal = {IEEE Trans. Inf. Theor.},
eprint={2208.05537},
archivePrefix={arXiv},
month = aug,
pages = {5100–5115},
numpages = {16}
}

@misc{sayginel2025faulttolerantlogicalcliffordgates,
      title={Fault-Tolerant Logical Clifford Gates from Code Automorphisms}, 
      author={Hasan Sayginel and Stergios Koutsioumpas and Mark Webster and Abhishek Rajput and Dan E Browne},
      year={2025},
      eprint={2409.18175},
      archivePrefix={arXiv},
      primaryClass={quant-ph},
      url={https://arxiv.org/abs/2409.18175}, 
}

@article{shor_9qubit_code,
  title = {Scheme for reducing decoherence in quantum computer memory},
  author = {Shor, Peter W.},
  journal = {Phys. Rev. A},
  volume = {52},
  issue = {4},
  pages = {R2493--R2496},
  numpages = {0},
  year = {1995},
  month = {Oct},
  publisher = {American Physical Society},
  doi = {10.1103/PhysRevA.52.R2493},
  url = {https://link.aps.org/doi/10.1103/PhysRevA.52.R2493}
}

@article{Kitaev_1997,
doi = {10.1070/RM1997v052n06ABEH002155},
url = {https://dx.doi.org/10.1070/RM1997v052n06ABEH002155},
year = {1997},
month = {Dec},
publisher = {},
volume = {52},
number = {6},
pages = {1191},
author = {A Yu Kitaev},
title = {Quantum computations: algorithms and error correction},
journal = {Russian Mathematical Surveys}
}

@inbook{freedman_meyer_luo,
title = "Z2-systolic freedom and quantum codes",
author = "Freedman, \{Michael H.\} and Meyer, \{David A.\} and Feng Luo",
note = "Publisher Copyright: {\textcopyright} 2002 by Chapman \& Hall/CRC.",
year = "2002",
month = jan,
day = "1",
isbn = "1584882824",
pages = "287--320",
booktitle = "Mathematics of Quantum Computation",
publisher = "CRC Press",
doi = {https://doi.org/10.1201/9781420035377}
}

@article{lin2023quantumtwoblockgroupalgebra,
  title = {Quantum two-block group algebra codes},
  author = {Lin, Hsiang-Ku and Pryadko, Leonid P.},
  journal = {Phys. Rev. A},
  volume = {109},
  issue = {2},
  pages = {022407},
  numpages = {17},
  year = {2024},
  month = {Feb},
  eprint={2306.16400},
  archivePrefix={arXiv},
  publisher = {American Physical Society},
  doi = {10.1103/PhysRevA.109.022407},
  url = {https://link.aps.org/doi/10.1103/PhysRevA.109.022407}
}

@misc{scruby2024highthresholdlowoverheadsingleshotdecodable,
      title={High-threshold, low-overhead and single-shot decodable fault-tolerant quantum memory}, 
      author={Thomas R. Scruby and Timo Hillmann and Joschka Roffe},
      year={2024},
      eprint={2406.14445},
      archivePrefix={arXiv},
      primaryClass={quant-ph},
      url={https://arxiv.org/abs/2406.14445}, 
}

@article{higgot_hgp,
  title = {Improved Single-Shot Decoding of Higher-Dimensional Hypergraph-Product Codes},
  author = {Higgott, Oscar and Breuckmann, Nikolas P.},
  journal = {PRX Quantum},
  volume = {4},
  issue = {2},
  pages = {020332},
  numpages = {17},
  year = {2023},
  month = {May},
  publisher = {American Physical Society},
  eprint={2206.03122},
  archivePrefix={arXiv},  
  doi = {10.1103/PRXQuantum.4.020332},
  url = {https://link.aps.org/doi/10.1103/PRXQuantum.4.020332}
}

@misc{koukoulekidis2_gb,
      title={Small Quantum Codes from Algebraic Extensions of Generalized Bicycle Codes}, 
      author={Nikolaos Koukoulekidis and Fedor Šimkovic IV and Martin Leib and Francisco Revson Fernandes Pereira},
      year={2024},
      eprint={2401.07583},
      archivePrefix={arXiv},
      primaryClass={quant-ph},
      url={https://arxiv.org/abs/2401.07583}, 
}

@article{Breuckmann_londe,
   title={Single-Shot Decoding of Linear Rate {LDPC} Quantum Codes With High Performance},
   volume={68},
   ISSN={1557-9654},
   url={http://dx.doi.org/10.1109/TIT.2021.3122352},
   DOI={10.1109/tit.2021.3122352},
   number={1},
   eprint={2001.03568},
   archivePrefix={arXiv},
   journal={IEEE Transactions on Information Theory},
   publisher={Institute of Electrical and Electronics Engineers (IEEE)},
   author={Breuckmann, Nikolas P. and Londe, Vivien},
   year={2022},
   month=jan, pages={272–286} }

@misc{grospellier2019numericalstudyhypergraphproduct,
      title={Numerical study of hypergraph product codes}, 
      author={Antoine Grospellier and Anirudh Krishna},
      year={2019},
      eprint={1810.03681},
      archivePrefix={arXiv},
      primaryClass={quant-ph},
      url={https://arxiv.org/abs/1810.03681}, 
}

@misc{kalachev2023twosidedrobustlytestablecodes,
      title={Two-sided Robustly Testable Codes}, 
      author={Gleb Kalachev and Pavel Panteleev},
      year={2023},
      eprint={2206.09973},
      archivePrefix={arXiv},
      primaryClass={cs.IT},
      url={https://arxiv.org/abs/2206.09973}, 
}

@inproceedings{dinur2022goodquantumldpccodes,
author = {Dinur, Irit and Hsieh, Min-Hsiu and Lin, Ting-Chun and Vidick, Thomas},
title = {Good Quantum LDPC Codes with Linear Time Decoders},
year = {2023},
isbn = {9781450399135},
publisher = {Association for Computing Machinery},
address = {New York, NY, USA},
url = {https://doi.org/10.1145/3564246.3585101},
doi = {10.1145/3564246.3585101},
booktitle = {Proceedings of the 55th Annual ACM Symposium on Theory of Computing},
pages = {905–918},
eprint={2206.07750},
archivePrefix={arXiv},
numpages = {14},
location = {Orlando, FL, USA},
series = {STOC 2023}
}

@misc{leverrier2022efficientdecodingconstantfraction,
      title={Efficient decoding up to a constant fraction of the code length for asymptotically good quantum codes}, 
      author={Anthony Leverrier and Gilles Zémor},
      year={2022},
      eprint={2206.07571},
      archivePrefix={arXiv},
      primaryClass={quant-ph},
      url={https://arxiv.org/abs/2206.07571}, 
}

@article{google_surface_code,
   title={Quantum error correction below the surface code threshold},
   volume={638},
   ISSN={1476-4687},
   url={http://dx.doi.org/10.1038/s41586-024-08449-y},
   DOI={10.1038/s41586-024-08449-y},
   number={8052},
   journal={Nature},
   eprint = {2408.13687},
   archivePrefix={arXiv},
   publisher={Springer Science and Business Media LLC},
   author={Rajeev Acharya and Dmitry A. Abanin and Laleh Aghababaie-Beni and Igor Aleiner and others},
   year={2024},
   month=dec, pages={920–926} }

@article{Zhao_2022,
   title={Realization of an Error-Correcting Surface Code with Superconducting Qubits},
   volume={129},
   ISSN={1079-7114},
   url={http://dx.doi.org/10.1103/PhysRevLett.129.030501},
   DOI={10.1103/physrevlett.129.030501},
   number={3},
   eprint = {2112.13505},
   archivePrefix={arXiv},
   journal={Physical Review Letters},
   publisher={American Physical Society (APS)},
   author={Zhao, Youwei and Ye, Yangsen and Huang, He-Liang and Zhang, Yiming and Wu, Dachao and Guan, Huijie and Zhu, Qingling and Wei, Zuolin and He, Tan and Cao, Sirui and Chen, Fusheng and Chung, Tung-Hsun and Deng, Hui and Fan, Daojin and Gong, Ming and Guo, Cheng and Guo, Shaojun and Han, Lianchen and Li, Na and Li, Shaowei and Li, Yuan and Liang, Futian and Lin, Jin and Qian, Haoran and Rong, Hao and Su, Hong and Sun, Lihua and Wang, Shiyu and Wu, Yulin and Xu, Yu and Ying, Chong and Yu, Jiale and Zha, Chen and Zhang, Kaili and Huo, Yong-Heng and Lu, Chao-Yang and Peng, Cheng-Zhi and Zhu, Xiaobo and Pan, Jian-Wei },
   year={2022},
   month={jul }}

@article{Krinner_2022,
   title={Realizing repeated quantum error correction in a distance-three surface code},
   volume={605},
   ISSN={1476-4687},
   url={http://dx.doi.org/10.1038/s41586-022-04566-8},
   DOI={10.1038/s41586-022-04566-8},
   number={7911},
   journal={Nature},
   eprint={2112.03708},
   archivePrefix={arXiv},
   publisher={Springer Science and Business Media LLC},
   author={Krinner, Sebastian and Lacroix, Nathan and Remm, Ants and Di Paolo, Agustin and Genois, Elie and Leroux, Catherine and Hellings, Christoph and Lazar, Stefania and Swiadek, Francois and Herrmann, Johannes and Norris, Graham J. and Andersen, Christian Kraglund and Müller, Markus and Blais, Alexandre and Eichler, Christopher and Wallraff, Andreas},
   year={2022},
   month=may, pages={669–674} }

@article{Weyl1912,
author = {Weyl, H.},
journal = {Mathematische Annalen},
pages = {441-479},
title = {Das asymptotische Verteilungsgesetz der Eigenwerte linearer partieller Differentialgleichungen (mit einer Anwendung auf die Theorie der Hohlraumstrahlung)},
url = {http://eudml.org/doc/158545},
volume = {71},
year = {1912},
}

\appendix 

\section{Expansion properties of left-right Cayley complexes}\label{Appendix0}

We examine more closely the spectral expansion properties of left-right Cayley complexes in terms of related Cayley graphs. Let $G$ be a group and let $A, B \subset G$ be sets of generators of $G$. Further, let $Cay(G, A)$ and $Cay(G, B)$ be Cayley graphs based on the left action of $A$ and the right action of $B$, respectively, on elements of $G$. Then the second largest eigenvalue of the left-right Cayley complex, when viewed as a graph, is closely related to the minimal second largest eigenvalue of $Cay(G, A)$ and $Cay(G, B)$.

\begin{lemma}
    Let $G$ be a finite group and let $A, B \subset G$ such that $\langle A, B \rangle = G$ and $|A| = |B| = \Delta$. Let $Cay(G, A)$ and $Cay(G, B)$ denote the Cayley graphs based on left and right actions on $G$, respectively. Let $\lambda_2^A$ and $\lambda_2^B$ denote the second largest eigenvalues of their respective adjacency matrices. Let $\Gamma(G, A,B)$ denote the left-right Cayley complex defined as in Def. \ref{def_LRCC}. Then we have
    \begin{equation*}
        \lambda_2^{LRCC} \leq \Delta + \min\{\lambda_2^a, \lambda_2^b\},
    \end{equation*}
    where $\lambda_2^{LRCC}$ denotes the second largest eigenvalue of $\Gamma(G, A, B)$. 
\end{lemma}

\begin{proof}
    Let $\mathbf{A}_A$ and $\mathbf{A}_B$ denote the adjacency matrices of $Cay(G, A)$ and $Cay(G, B)$, respectively. Then $\text{spec}(\mathbf{A}_A) = \{\lambda_1^a, \lambda_2^a, \ldots, \lambda_n^a \} $ and $\text{spec}(\mathbf{A}_B) = \{\lambda_1^b, \lambda_2^b, \ldots, \lambda_n^b \} $, where $\Delta  = \lambda_1^m > \lambda_2^m > \ldots > \lambda_n^m$ for $m \in \{a, b\}$ and $n = |G|$. 

    Consider the graph $X(G, A, B)$ which has vertices corresponding to a single copy of the elements of $G$ and edges of the form $(g, ag)$ and $(g, gb)$. The adjacency matrix of $X(G, A, B)$ is then $\mathbf{C} = \mathbf{A}_A + \mathbf{A}_B$. Note that $\mathbf{A}_A$ and $\mathbf{A}_B$ have no overlapping entries due to the total non-conjugacy condition. Due to Weyl \cite{Weyl1912}, we have 
    \begin{equation*}
        \lambda_{i + j -1}^C \leq \lambda_{i}^a + \lambda_j^b,
    \end{equation*}
    since $\mathbf{A}_A$ and $\mathbf{A}_B$ are both Hermitian as symmetric matrices containing only real entries. This, in turn, means that 
    \begin{equation*}
        \lambda_2^C \leq \min \{\lambda_1^a + \lambda_2^b, \lambda_2^a + \lambda_1^b \} = \Delta + \min \{ \lambda_2^a, \lambda_2^b \}.
    \end{equation*}
    Finally, as the bipartite double cover of $X(G, A, B)$, the adjacency matrix of $\Gamma(G, A, B)$ is of the form 
    \begin{equation*}
        \mathbf{A}^{\Gamma} = \begin{pmatrix}
            0 & C \\ C & 0
        \end{pmatrix}.
    \end{equation*}
It is easy to show that $\text{spec}(\mathbf{A}^{\Gamma}) = \{\pm \lambda  \; |\; \lambda \in \text{spec}(C) \}$, so we can conclude that 
    \begin{equation*}
        \lambda_2^{LRCC} \leq \Delta + \min \{ \lambda_2^a, \lambda_2^b \}.
    \end{equation*}
\end{proof}

\section{Restrictions on $\Delta$ values for the construction of LRCCs}\label{AppendixA}

There are a number of general and group-specific limitations on the $\Delta$ values which can be used to construct left-right Cayly complexes that fulfill the total non-conjugacy condition from Def. \ref{tnc def}. In particular, LRCCs constructed based on dihedral groups $D_n$ with $n$ odd and odd values of $\Delta$ violate the TNC condition.

\begin{lemma}
    Let $G$ be a group and let $A, B \subset G$ such that $\langle A, B \rangle = G$ and $|A| = |B| = \Delta$. If $A \cap B \neq \emptyset $ then the LRCC $\Gamma(G, A, B)$ does not satisfy the total non-conjugacy condition.
\end{lemma}

\begin{proof}
    We have $A \cap B \neq \emptyset \implies \exists a \in A, b \in B : a = b \implies ae = a = b = be$
    for $e \in G$ the identity element. This violates the TNC condition. 
\end{proof}

\begin{corollary}
    In order to generate a valid LRCC which fulfills the TNC, we must have $\Delta < |G|/2$. 
\end{corollary}

\begin{proof}
    If $\Delta = |G|/2$, then $e$ is contained in either $A$ or $B$, which does not yield a valid Cayley graph $X(G,A)$ or $X(G, B)$, respectively. If $\Delta > |G|/2$, then $A \cap B \neq \emptyset $ and the TNC condition is violated according to the lemma above. 
\end{proof}

\begin{lemma}
    Let $n \in \mathbb{N}$ and $\Delta \in [|D_n|]$ be odd and let $A, B \subset D_n$ be symmetric such that $\langle A, B \rangle = D_n$. Then the LRCC $\Gamma(G, A, B)$ does not satisfy the total non-conjugacy condition.
\end{lemma}

\begin{proof}
    Since $A$ and $B$ are symmetric and of odd cardinality, there exist an element $a \in A$ such that $a = a^{-1}$ and $b \in B$ such that $b = b^{-1}$. Since $a$ and $b$ are elements of the dihedral group $D_n$, they are reflections. As a consequence of the Sylow theorem, $n$ odd implies that all reflections are conjugate to one another, meaning
    \begin{equation*}
        g^{-1}ag = b \; \forall g \in G,
    \end{equation*}
    which violates the TNC condition. 
\end{proof}

\section{Construction Details}\label{appendixB}
\begin{table*}[ht!]
\begin{center}
\begin{tabular}{|| c | c | c| c | c | c | c ||}
 \hline
$[[n, k, d]]$ & Group & $\Delta$ & $A$ & $B$& $C_A$ & $C_B$ \\
\hline

[[36, 8, 3]] & $ D_4$ & 3 & $\{s, r, r^3 \} $& $\{sr, sr^3, r^2 \}$ & $\begin{bmatrix}
   1 & 1 & 1
\end{bmatrix}$  & $\begin{bmatrix}
    
     1 & 0 & 0 \\ 1 & 1 & 1
\end{bmatrix} $  \\
\hline

[[54, 11, 4]] & $D_6$ & 3&  $\{r, r^3, r^5 \}$ & $\{ sr^2, sr^4, sr^5 \}$   &  
$\begin{bmatrix}
   1 & 1 & 1
\end{bmatrix}$  & $\begin{bmatrix}
     1 & 0 & 0 \\ 1 & 1 & 1
\end{bmatrix} $  \\
\hline

[[72, 14, 4]] & $D_8$ & 3 & $\{s, sr^4, r^4 \}$& $\{sr, sr^3, sr^7 \}$  &
$\begin{bmatrix}
   1 & 1 & 1
\end{bmatrix}$  & $\begin{bmatrix}
    
     1 & 0 & 0 \\ 1 & 1 & 1
\end{bmatrix} $  \\
\hline

[[200, 10, 10]] & $D_8$& 5 & $\{sr^6, r, r^3, r^5, r^7\}$ & $\{ sr, sr^3, sr^7, r^2, r^6\}$ & $\begin{bmatrix}
    1 & 0 & 0 & 1 & 1 \\
    0 & 0& 1 &1 & 1 \\
\end{bmatrix}$&  

$\begin{bmatrix}
    1 & 0& 0& 0 & 1\\
    0 & 0 & 0 & 1  &1 \\
    1 & 1& 1& 0 & 0\\
\end{bmatrix}$\\
\hline

[[250, 10, 15]] & $D_{10}$ & 5 & $\{sr, r, r^3, r^7, r^9, \}$ & $ \{sr^6, r^2, r^4,  r^6, r^8 \}$&

$\begin{bmatrix}
    1& 0& 0& 1 & 1\\
    0 & 1& 1 & 1 & 1 \\
\end{bmatrix}$&

$\begin{bmatrix}
    1 & 0 & 0 & 0 & 1\\
    0 & 0& 0 & 1 & 1 \\
    1 & 0 & 1 & 0 & 1
\end{bmatrix}$

\\
\hline

\end{tabular}
\end{center}
\caption{Explicit instances of quantum Tanner codes are listed along with the specifications of their constructions. We provide the sets of generators $A$ and $B$, whose elements are given in terms of the standard representation of dihedral groups $D_n = \langle r, s \; | \; s^2 = r^n = e, srs = r^{-1} \rangle$. The last two columns contain the generator matrices of the classical codes $C_A$ and $C_B$, respectively, which were generated randomly. }
\label{construction_params}
\end{table*}

We provide the specific generating sets and classical code pairs for the quantum Tanner code constructions in Table \ref{params}. Thousands of random instances were generated, and the best instances, in terms of encoding rate and distance, are listed here. We use the standard representation of dihedral groups, where $D_n = \langle r, s \; | \; s^2 = r^n = e, srs = r^{-1} \rangle$. The symmetric sets of generators, $A$ and $B$, as well as the resulting left-right-Cayley complexes were generated using $\texttt{GAP}$ \cite{GAP4}. 

The classical codes are given by their parity check matrices, whose rows represent parity checks and whose columns correspond to bits. These were generated randomly using \texttt{Julia}. The quantum Tanner codes were then constructed in \texttt{Julia} using the \texttt{LinearAlgebra} and \texttt{AbstractAlgebra} packages. Distances were computed and verified using the \texttt{LDPC} Python package and the \texttt{GAP} package \texttt{QDistRnd}.

\end{document}